\documentclass[prl,aps,reprint,superscriptaddress,nofootinbib]{revtex4-1}

\usepackage[utf8]{inputenc}
\usepackage[english]{babel}

\usepackage{amssymb}
\usepackage{amsmath}
\usepackage{amsfonts}
\usepackage{mathtools}
\usepackage{amsthm}
\usepackage{graphicx}
\usepackage{float}
\usepackage{dsfont} 
\usepackage{cancel}
\usepackage{enumitem}
\usepackage{appendix}
\usepackage{simplewick}
\usepackage{bbm}
\usepackage{xcolor}
\usepackage{hyperref}

\usepackage{soul}
\usepackage{comment}

\newcommand{\be}{\begin{equation}}
\newcommand{\ee}{\end{equation}}
\newcommand{\bea}{\begin{eqnarray}}
\newcommand{\eea}{\end{eqnarray}}
\newcommand\nn{\nonumber}

\makeatletter

\theoremstyle{plain}
\newtheorem{thm}{\protect\theoremname}
\theoremstyle{remark}

\theoremstyle{plain}
\newtheorem{lem}[thm]{\protect\lemmaname}
\theoremstyle{plain}
\newtheorem{cor}[thm]{\protect\corollaryname}
\theoremstyle{definition}

\makeatother

\addto\captionsamerican{\renewcommand{\corollaryname}{Corollary}}
\addto\captionsamerican{\renewcommand{\definitionname}{Definition}}
\addto\captionsamerican{\renewcommand{\lemmaname}{Lemma}}
\addto\captionsamerican{\renewcommand{\remarkname}{Remark}}
\addto\captionsamerican{\renewcommand{\theoremname}{Theorem}}
\addto\captionsenglish{\renewcommand{\corollaryname}{Corollary}}
\addto\captionsenglish{\renewcommand{\definitionname}{Definition}}
\addto\captionsenglish{\renewcommand{\lemmaname}{Lemma}}
\addto\captionsenglish{\renewcommand{\remarkname}{Remark}}
\addto\captionsenglish{\renewcommand{\theoremname}{Theorem}}
\providecommand{\corollaryname}{Corollary}
\providecommand{\definitionname}{Definition}
\providecommand{\lemmaname}{Lemma}
\providecommand{\remarkname}{Remark}
\providecommand{\theoremname}{Theorem}

\makeatletter

\begin{document}

\title{Quantum Complementarity through Entropic Certainty Principles}
\begin{flushright}
\end{flushright}

\author{Javier M. Magan and Diego Pontello}
\affiliation{Instituto Balseiro, Centro Atomico Bariloche S. C. de Bariloche, Rio Negro, R8402AGP, Argentina}

\begin{abstract}
We approach the physical implications of the non-commutative nature of Complementary Observable Algebras (COA) from an information theoretic perspective. In particular, we derive a general \textit{entropic certainty principle} stating that the sum of two relative entropies, naturally related to the COA, is equal to the so-called algebraic index of the associated inclusion. Uncertainty relations then arise by monotonicity of the relative entropies that participate in the underlying entropic certainty. Examples and applications are described in quantum field theories with global symmetries, where the COA are formed by the charge-anticharge local operators (intertwiners) and the unitary representations of the symmetry group (twists), and in theories with local symmetries, where the COA are formed by Wilson and 't Hooft loops. In general, the entropic certainty principle naturally captures the physics of order/disorder parameters, a feature that makes it a generic handle for the information theoretic characterization of quantum phases.

\end{abstract}

\maketitle


\section{The Uncertainty Principle. Old and New.} 

The uncertainty principle limits the precision of potential experiments in a quantum mechanical world. It states that the product of statistical errors is bounded away from zero. The paradigmatic example is that of position/momentum operators
\begin{equation}
\sigma_{x}\sigma_{p}\geqslant \hbar/2\,,
\end{equation}
where $\sigma$ stands for the variance. Such bounds are derived and improved by allowing certain state dependence. In its generalized form, applicable to any pair of observables $A$ and $B$, it reads
\begin{equation}
\sigma_{A}^{2}\sigma_{B}^{2}\geqslant \vert \frac{1}{2}\langle\lbrace A, B\rbrace\rangle -\langle A\rangle\langle B\rangle\vert^{2}+\vert\langle\frac{1}{2i} [A,B]\rangle\vert^{2}\,.
\end{equation}
But the implications of quantum complementarity can be seen through different glasses as well. The most natural one is the entropy, and the emerging relations are called entropic uncertainty relations, see \cite{Bialynicki-Birula2011,Wehner_2010,RevModPhysColes} for recent reviews with a historical account. A famous example is \cite{PhysRevLett.60.1103}
\begin{equation}
H(A)+H(B)\geqslant -\log c \,,
\end{equation}
where $H(A)$ and $H(B)$ are the classical Shannon entropies associated to the measurements outcomes of observables $A$ and $B$ in a $d$-dimensional Hilbert space, and 
\be
c:=\mathrm{max} \{ \left|\left\langle \psi_{i}^{A}\mid\psi_{j}^{B}\right\rangle \right|\,:\,i,j=1,\ldots,d \}\;,
\ee
is the maximum overlap between any two eigenvectors $\vert\psi_{i}^{A}\rangle$ and $\vert\psi_{j}^{B}\rangle$ of $A$ and $B$.

Entropic approaches are motivated for different reasons. First, in some cases, entropic uncertainty relations are stronger than the conventional ones \cite{Bialynicki-Birula2011}. Second, they have a precise operational meaning in the context of information theory \cite{journals/bstj/Shannon48}. Finally, they suggest further generalizations of the uncertainty principle, such as including memories (side-information), to which the observer has access too. A well-studied example is \cite{PhysRevLett.103.020402,PhysRevLett.108.210405}
\be
H(A\vert M)+ H(B\vert M)\geqslant -\log c \,,
\ee
where $H(A\vert M)$ and $H(B\vert M)$ are the conditional entropies with respect to the memory $M$. From the previous relations, it can also be derived
\begin{equation}\label{mi}
I(A,M)+I(B,M)\leqslant \log \left(d^2 c \right)\,,
\end{equation}
where $I(A,B) := S(A)+S(B)-S(A \cup B)$ stands for the mutual information. Intuitively, the uncertainty principle constrains the amount of information that $M$ can have about a certain observable $A$, given the amount it has about a complementary one $B$. A nice way to prove this relation uses the monotonicity of the relative entropy \cite{PhysRevLett.108.210405}. Below, we use the monotonicity of the relative entropy in a different way to obtain a large family of uncertainty relations.

In this letter, we explore two novel avenues or consequences of quantum complementarity. On one hand, instead of two observables, we consider two Complementary Operator Algebras (COA) $\mathcal{A}$ and $\mathcal{B}$. These are two observable subalgebras containing at least some operators which do not commute with each other. These algebras can be considered the ``laboratories'' (specifying the set of allowed experiments) associated with different observers. On the other hand, instead of inequalities, we seek to find  ``entropic certainties''. This is inspired by recent results in the context of QFT, where the first example of an entropic certainty principle was discovered \cite{Casini:2019kex}. Such an example, together with other applications, is described in the last section.

\section{Relative entropy and Complementary Observable Algebras.} 

Instead of entropies, our principle will use a different information measure: the relative entropy. This quantity measures the distinguishability between two states with respect to a given algebra $\mathcal{M}$. For a finite dimensional algebra $\mathcal{M}$, it is defined by
\be
S_{\mathcal{M}}\left(\omega\mid\phi\right):=\mathrm{Tr}_{\mathcal{M}}\left(\rho^\omega\left(\log\rho^\omega-\log\rho^\phi\right)\right)\,,\label{re_def}
\ee
where $\mathrm{Tr}_{\mathcal{M}}$ is the canonical trace on $\mathcal{M}$, and $\rho^\omega$, $\rho^\phi$ are the density matrices representing the underlying states (see \cite{ohya1993quantum} or appendix A).

Using relative entropy is convenient from several perspectives. First, all other information quantities can be derived from it. For example, for a full matrix algebra acting on a finite $d$-dimensional Hilbert space the entropy is
\be
S_{\mathcal{M}} (\omega) := -\mathrm{Tr}_{\mathcal{M}}\left(\omega\log\omega\right) =\log d - S_{\mathcal{M}}\left(\omega\mid \tau\right)\,, \label{sre}
\ee
where $\tau:=\mathbf{1}/d$ is the maximally mixed density matrix. Second, relative entropy shows monotonicity under general quantum channels and restrictions onto subalgebras \cite{ohya1993quantum}. Finally, relative entropy is well-defined across different types of algebras, including type $III$ von Neumann algebras appearing in QFT.

To motivate our entropic certainty principle, notice  that in \eqref{sre} the maximally mixed state $\tau$ can be defined by composing $\omega$ with a map $\varepsilon:\mathcal{M}\rightarrow \mathds{1}$, defined by $\varepsilon (m) := \tau (m)\mathbf{1}$. This rewriting leads to
\be
S_{\mathcal{M}} (\omega)=\log d - S_{\mathcal{M}}(\omega\mid \omega\circ \varepsilon)\,,
\ee
and suggests further generalizations. First, the map $\varepsilon:\mathcal{M}\rightarrow \mathds{1}$ is one example of a whole space of such maps. Besides, instead of the identity as the target algebra, we could choose any subalgebra $\mathcal{N}\subset\mathcal{M}$. The maps $\varepsilon : \mathcal{M}\rightarrow \mathcal{N}$ are called conditional expectations \cite{ohya1993quantum}. They are positive, linear, and unital maps from an algebra $\mathcal{M}$ to a subalgebra $\mathcal{N}$ satisfying
\be
\hspace{-1mm} \varepsilon\left(n_{1}\,m\,n_{2}\right)=n_{1}\varepsilon\left(m\right)n_{2}\,,\hspace{3mm} \forall m\in\mathcal{M},\,\forall n_{1},n_{2}\in\mathcal{N}.\label{ce_def_prop}
\ee
These maps are the mathematical definition of what restricting our observational abilities means. Examples are tracing out part of the system or retaining the neutral part of a subalgebra under the action of a certain symmetry group. If $\mathcal{M}= \mathcal{N}\vee \mathcal{A}$ is the algebra generated by $\mathcal{N}$ and certain algebra $\mathcal{A}$, we say the conditional expectation ``kills'' $\mathcal{A}$. Notice also that the conditional expectation can be used to lift a state $\omega_{\mathcal{N}}$ in $\mathcal{N}$ to a state $\omega=\omega_{\mathcal{N}}\circ\varepsilon$ in $\mathcal{M}$.

In this context $S_{\mathcal{M}}(\omega \mid \omega\circ\varepsilon)$ measures the amount of information in $\mathcal{A}$ which is erased by $\varepsilon$. It takes into account side correlations with $\mathcal{N}$, which is the subalgebra left invariant by $\varepsilon$. In this light, complementarity refers to the algebra $\mathcal{A}$, while $\mathcal{N}$ plays the role of quantum memory. What remains to be answered is who plays the role of the Complementary Observable Algebra (COA). A canonical candidate arises from the following diagram
\begin{eqnarray}
\mathcal{M} & \overset{\varepsilon}{\longrightarrow} & \mathcal{N}\nonumber \\
\updownarrow\prime\! &  & \:\updownarrow\prime\\ \label{ecr_diagr}
\mathcal{M}' & \overset{\varepsilon'}{\longleftarrow} & \mathcal{N}'\,.\nonumber 
\end{eqnarray}
In this diagram, going vertically takes the algebras to its commutants, while horizontally in the arrow direction means restricting to the target subalgebra. If $\varepsilon$ kills the algebra $\mathcal{A}\subset\mathcal{M}$, the \textit{dual conditional expectation} $\varepsilon '$ kills the COA $\tilde{\mathcal{A}}\subset\mathcal{N}'$. Notice that $\tilde{\mathcal{A}}$ does not commute with $\mathcal{A}$.

As an example, take $\mathcal{M}$ as the abelian algebra $\mathcal{X}$ generated by the position operator, and a conditional expectation that kills the full $\mathcal{M}=\mathcal{A} =: \mathcal{X}$. We then obtain
\bea
\mathcal{X} & \overset{\varepsilon}{\longrightarrow} & \mathds{1}\nonumber \\
\updownarrow\prime \!\! &  & \,\updownarrow\prime\\
\mathcal{X} & \overset{\varepsilon'}{\longleftarrow} & \mathcal{X}\vee\mathcal{P}\,.\nonumber 
\eea
We conclude that the COA of $\mathcal{X}$ is $\mathcal{P}$, the algebra generated by the momentum operator, as expected.

Having the COAs $\mathcal{A}$ and $\tilde{\mathcal{A}}$, we might expect an uncertainty relation of the form
\be
S_{\mathcal{M}}\left(\omega|\omega\circ\varepsilon\right)+S_{\mathcal{N}'}\left(\omega|\omega\circ\varepsilon'\right)\leq\log\lambda\,,
\ee
where $\lambda$ is a constant to be determined. This is analogous to equation \eqref{mi} but generalized in different ways. The difference is that, given $\varepsilon$, we have some freedom when choosing $\varepsilon '$. In this letter, we show we can choose it so as to obtain the \emph{entropic certainty principle}
\be
S_{\mathcal{M}}\left(\omega|\omega\circ\varepsilon\right)+S_{\mathcal{N}'}\left(\omega|\omega\circ\varepsilon'\right)=\log\lambda \,, \label{cerp}
\ee
where $\lambda$ is a fixed number and $\omega$ an arbitrary pure state. This is a generalization of the known equality of entanglement entropy associated to commutant observable algebras in a global pure state to non-commutative COA.

\section{The space of conditional expectations.}  \label{section_ce}

To approach this problem, we need to understand the space of conditional expectations $C(\mathcal{M},\mathcal{N})$ for a generic inclusion of algebras $\mathcal{N}\subset\mathcal{M}$. This was studied in \cite{umegaki1,umegaki2,umegaki3,umegaki4}, but here we use a different approach. Irrespective of the inclusion, both algebras have the following general form
\be
\hspace{-1.8mm} \mathcal{M}  \cong  \bigoplus_{j=1}^{z_{\mathcal{M}}}M_{m_{j}}(\mathbb{C})\otimes\mathds{1}_{m'_{j}}, \hspace{2.2mm}
\mathcal{\mathcal{N}} \cong  \bigoplus_{k=1}^{z_{\mathcal{N}}}M_{n_{k}}(\mathbb{C})\otimes\mathds{1}_{n'_{k}},\label{m_inc}
\ee
where $M_{m}(\mathbb{C})$ is the full matrix algebra of $m\times m$ complex matrices. The minimal central projectors of the algebra $\mathcal{M}$ are $P_{j}^{\mathcal{M}}:=\mathbf{1}_{m_{j}}\otimes\mathbf{1}_{m'_{j}}$, with $j=1,\ldots,z_{\mathcal{M}}$. And similarly for the algebra $\mathcal{N}$. We define the subalgebras $\mathcal{M}_{k} := P_{k}^{\mathcal{N}}\mathcal{M}P_{k}^{\mathcal{N}}$.

Given these algebras, generic inclusions $\mathcal{N}\subset\mathcal{M}$ are characterized by the number of times $\mu_{kj}$ the factor $M_{n_{k}}(\mathbb{C})$ of $\mathcal{N}$ is included in the factor $M_{m_{j}}(\mathbb{C})$ of $\mathcal{M}$ \cite{Jones1983,teruya,giorlongo,Giorgetti:2018nji}. They must satisfy $m_{j} = \sum_{k=1}^{z_{\mathcal{N}}}n_{k}\times\mu_{kj}\,$.

To understand the space $C(\mathcal{M},\mathcal{N})$ it is useful to split $\mathcal{N}\subset\mathcal{M}$ in two steps
\begin{equation}
\mathcal{M}\supset\bigoplus_{k=1}^{z_{\mathcal{N}}}M_{n_{k}}\left(\mathbb{C}\right)\otimes\left(\bigoplus_{j=1}^{z_{\mathcal{M}}}M_{\mu_{kj}}(\mathbb{C})\otimes\mathds{1}_{m'_{j}}\right)\supset\mathcal{N}\,,
\end{equation}
where we must also have $n'_{k}=\sum_{j=1}^{z_{\mathcal{M}}}\mu_{kj}\times m'_{j}$. The hint towards $C(\mathcal{M},\mathcal{N})$ is that any $\varepsilon \in C(\mathcal{M},\mathcal{N})$ arises by composing a conditional expectation from $\mathcal{M}$ to the intermediate algebra, which turns out to be unique, and a conditional expectation from such intermediate algebra to $\mathcal{N}$, whose space can be simply derived. The following lemma determines $C(\mathcal{M},\mathcal{N})$ completely.

\begin{lem} \label{lemma_cem}
Let $\mathcal{N}\subset\mathcal{M}$ be a generic inclusion of finite dimensional algebras, and $\varepsilon\in C(\mathcal{M},\mathcal{N})$ a conditional expectation. Then, there exist unique conditional expectations  $\varepsilon_k\in C(\mathcal{M}_k,\mathcal{N}_k)$ such that
\be
\varepsilon(A) =\bigoplus_{k=1}^{z_{\mathcal{N}}}\varepsilon_{k}(A_{k})\,, \hspace{5mm} A_{k}:=P_{k}^{\mathcal{N}}AP_{k}^{\mathcal{N}} \in \mathcal{M}_k \,.
\ee
These conditional expectations are uniquely determined by states $\rho_{k}^{\varepsilon}=\bigoplus_{j=1}^{z_{\mathcal{M}}}p_{kj}^{\varepsilon}\,\rho_{kj}^{\varepsilon}$ on $\mathcal{M}_{k}\cap\mathcal{N}'\simeq\bigoplus_{j=1}^{z_{\mathcal{M}}}M_{\mu_{kj}}(\mathbb{C})$ (with $\sum_{j=1}^{z_{\mathcal{M}}}p_{kj}^{\varepsilon}=1$) through
\be
\varepsilon_{k}(A_{k}) :=\left(\sum_{j=1}^{z_{\mathcal{M}}}p_{kj}^{\varepsilon}\,\mathrm{Tr}\left(\rho_{kj}^{\varepsilon}C_{kj}\right)\right)\left(B_{k}\otimes\mathbf{1}_{n'_{k}}\right) \, , \label{par_ce}
\ee
where $A_{k}:= B_{k}\otimes\left(\bigoplus_{j=1}^{z_{\mathcal{M}}}C_{kj}\otimes\mathbf{1}_{m'_{j}}\right)$, $B_{k}\in M_{n_{k}}(\mathbb{C})$, and $C_{kj}\in M_{\mu_{kj}}(\mathbb{C})$. Equation \eqref{par_ce} is extended to a general element $A_{k}\in\mathcal{M}_{k}$ by linearity.
\end{lem}
\begin{proof}
See appendix A.
\end{proof}
Therefore, the conditional expectations $\varepsilon$ and $\varepsilon '$ are parametrized by  states $\rho_{k}^{\varepsilon}$ on $\mathcal{M}_{k}\cap\mathcal{N}'$, and $\tilde{\rho}_{j}^{\varepsilon}$ on $\mathcal{N}'_{j}\cap\mathcal{M}$ respectively, where $\mathcal{N}'_{j}:=P_{j}^{\mathcal{M}}\mathcal{N'}P_{j}^{\mathcal{M}}$. 

Before we finish this section, let us make some definitions. Given $\varepsilon$ parametrized by $\rho_{k}^{\varepsilon}$, let $\{ t_{l,jk} \,:\,l=1,\ldots,\mu_{kj}\}$ be the eigenvalues of $\rho_{kj}^{\varepsilon}$. Then, we define
\be
\lambda_{jk}(\varepsilon)  :=  \sum_{l=1}^{\mu_{kj}}\frac{1}{t_{l,jk}}\quad \mathrm{and} \quad
\lambda_{j}(\varepsilon)  :=  \sum_{k=1}^{z_{\mathcal{N}}}\frac{\lambda(\varepsilon_{jk})}{p_{kj}^{\varepsilon}}\,.\label{lf} 
\ee
Besides, we define the space $\hat{C}\left(\mathcal{M},\mathcal{N}\right)\subset C\left(\mathcal{M},\mathcal{N}\right)$ to be the one formed by those conditional expectations for which $\lambda_{j}(\varepsilon):=\lambda(\varepsilon)$ is independent of $j$.

\section{Entropic certainties and the algebraic index.}

The following theorem, which is the main result of the present letter, holds for ``connected inclusions" $\mathcal{N}\subset\mathcal{M}$, which satisfy the extra property $\mathcal{Z}(\mathcal{M})\cap\mathcal{Z}(\mathcal{N})=\mathcal{\mathds{1}}$ \cite{teruya}. Non-connected inclusions are direct sums of connected ones, and \eqref{cerp} holds for each term of the sum independently. We consider them in appendix B.

\begin{thm} \label{cer_thm}
Let $\mathcal{N}\subset\mathcal{M}\subset\mathcal{B}\left(\mathcal{H}\right)$
be a connected inclusion of finite dimensional algebras. Then, for every $\varepsilon\in\hat{C}\left(\mathcal{M},\mathcal{N}\right)$
there exists a unique $\varepsilon'\in\hat{C}\left(\mathcal{M},\mathcal{N}\right)$ such that
\begin{equation}
S_{\mathcal{M}}\left(\omega|\omega\circ\varepsilon\right)+S_{\mathcal{N}'}\left(\omega|\omega\circ\varepsilon'\right)=\log\left(\lambda\right)\,,\label{cer_eq}
\end{equation}
holds for any global pure state $\omega$ on $\mathcal{B}\left(\mathcal{H}\right)$, and where $\lambda:=\lambda(\varepsilon)\equiv\lambda(\varepsilon')$ is the algebraic index of the conditional expectations, to be described below. Moreover, the conditional expectation $\varepsilon_{0}$ that minimizes $\lambda$ always exists, it is unique, and $\varepsilon_{0}\in \hat{C}\left(\mathcal{M},\mathcal{N}\right)$. The \textit{minimal index} $\lambda(\varepsilon_0) =: [\mathcal{M}:\mathcal{N}]$ only depends on the inclusion $\mathcal{N}\subset\mathcal{M}$.
\end{thm}
\begin{proof}
See appendix B.
\end{proof}

As stated, $\lambda \ge 1$, defined in \eqref{lf}, is the algebraic index of the conditional expectations. This index has been extensively studied in the mathematical and physics communities. The first notion was proposed by Jones in the context of inclusions of type $II_{1}$ subfactors \cite{Jones1983}. It was later noticed independently by Kosaki and Longo \cite{KOSAKI1986123,longo1989}, that the index was most naturally associated with a conditional expectation, and both were able to extend the definition to type $III$ algebras. We review the definition and several examples in appendix C.

Probably, the most notable application of index theory to quantum physics concerns a discovery by Longo \cite{longo1989} in the context of the algebraic approach to superselection sectors in QFT, developed by Doplicher, Haag and Roberts \cite{Doplicher:1969tk,Doplicher:1969kp,Doplicher:1973at,Haag:1992hx}. He found that the dimension $d_{r}$ of a superselection sector, associated to a representation $r$ of $G$ and characterized by an endomorphism $\rho_{r}$ of the observable algebra $\mathcal{O}$ is related to the minimal index of the inclusion $\rho_{r} (\mathcal{O})\subset\mathcal{O}$ by means of
\be
[\mathcal{O}:\rho_{r} (\mathcal{O})] =d_r^{2} \, .
\ee
Also, considering the field algebra $\mathcal{F}$, including all charged operators, and the observable (neutral respect to $G$) subalgebra $\mathcal{O}\subset \mathcal{F}$, one obtains \cite{Longo:1994xe} 
\be
[\mathcal{O}:\rho(\mathcal{O})]^\frac{1}{2}= [\mathcal{F}:\mathcal{O}] =\sum_{r}d_{r}^{2}\,,
\ee
where $\rho\simeq\bigoplus d_{r}\rho_{r}$, $r$ runs over the irreducible representations of $G$, and $d_{r}$ are their dimensions.

There is also an old application of index theory to information theory \cite{popa}. It works as follows. In the context of inclusion of subfactors, the index can be defined by the so-called Pimsner-Popa bound \cite{popa}
\be
\varepsilon (m^{+})\geq \lambda^{-1}\,m^{+}\,,\hspace{4mm} \forall m^{+}\in\mathcal{M}_{+} \,. \label{popa_eq}
\ee
In the original references \cite{KOSAKI1986123,longo1989}, Kosaki and Longo showed their index definitions imply the bound. This bound can be used to constrain relative entropies. We first notice that for two normal states $\omega$ and $\omega'$ on $\mathcal{M}$ satisfying $\omega\geq \mu \omega'$, we have that $S_{\mathcal{M}}\left(\omega|\omega'\right)\leq \log \mu^{-1}$ \cite{ohya1993quantum}. Therefore,
\be
S_{\mathcal{M}}\left(\omega|\omega\circ\varepsilon\right)\leq \log \lambda\,. \label{re_popa}
\ee
Further similar applications have appeared recently in \cite{Longo:2017mbg} and \cite{Naaijkens_2018}. Relation \eqref{re_popa} is explicit in the entropic certainty relation \eqref{cer_eq}, since $S_{\mathcal{N}'}\left(\omega|\omega\circ\varepsilon'\right)\geq 0$. But \eqref{cer_eq} further improves such bounds by unraveling the cause for the depart from saturation. 

\section{Applications.} 

We begin with two physical examples which have been discussed extensively in \cite{Casini:2019kex,Casini:ls}. First, we consider QFTs with global symmetries. These are both those theories for which $G$ acts equally in all space, and those for which the charged operators are local operators \cite{Haag:1992hx}. The existence of the unitary representation $U_{g}$ of $G$ suggests the existence of local representations $\tau_{g}^{A}$, where $A$ is the space domain of support of the operator, satisfying the group algebra
\be
\tau_{g}^{A}\tau_{g'}^{A}=\tau_{gg'}^{A}\,.
\ee
The construction of $\tau_{g}^{A}$ in the lattice is trivial. In the continuum is more subtle, but it can be done  \cite{doplicher1983,Doplicher:1984zz,Bueno:2020vnx}. The operators $\tau_{g}^{A}$ are called \textit{twists}.
If $G$ is non-abelian, $\tau_{g}^{A}$ are not invariant under $G$. We can construct invariant combinations by averaging over the group, resulting in one invariant twist $\tau_{[g]}^{A}$ per equivalence class of $G$ \cite{Casini:2019kex,Casini:ls}.

On the other hand, the local charged operators $V_{r,i}^{A}$, localized in region $A$ and transforming according to the irreducible representation $r$ as $\tau_{g}^{A}V_{r,i}^{A}(\tau_{g}^{A})^{-1}=\sum_{j=1}^{d_r} \mathcal{R}_{ij}^r(g)V_{r,j}^{A}$, suggest the construction of non-local neutral operators, called \textit{intertwiners}, formed by contracting two charged operators localized on different regions
\be
\mathcal{I}_{r}^{AB}:=\sum_{i=1}^{d_{r}}V_{r,i}^{A}\,V_{r,i}^{B\,\dagger}\,.
\ee
Crucially $[\mathcal{I}_{r},\tau_{[g]}^{A}]\neq 0$, since $\tau_{[g]}^{A}$ acts as a group transformation on $A$, but as the identity on $B$. Then, the meaningful COA is formed by the twists $\tau^{A}$ and the intertwiners $\mathcal{I}^{AB}$. Consider now two disjoint regions $A$ and $B$, with neutral (observable) algebras $\mathcal{O}_{A}$ and $\mathcal{O}_{B}$, and its complementary region $(AB)'$ with algebra $\mathcal{O}_{(AB)'}$. Defining $\mathcal{O}_{AB}\equiv\mathcal{O}_{A}\vee\mathcal{O}_{B}$ the COA diagram \eqref{ecr_diagr} reads
\bea
\mathcal{O}_{AB}\vee \mathcal{I}_{AB} & \overset{\varepsilon}{\longrightarrow} & \mathcal{O}_{AB}\nonumber \\
\updownarrow\prime\! &  & \:\updownarrow\prime\\
\mathcal{O}_{(AB)'} & \overset{\varepsilon'}{\longleftarrow} & \mathcal{O}_{(AB)'}\vee \tau^{A}\,,\nonumber 
\eea
where $\varepsilon$ kills the intertwiners and $\varepsilon '$ kills the twists. The minimal index appearing in the entropic certainty becomes
\be
\lambda_{\mathrm{min}}=\sum_{r}d_{r}^{2}=\left|G\right|\,,
\ee
which coincides with the expression for the topological entanglement entropy \cite{PhysRevLett.96.110404,PhysRevLett.96.110405}. This coincidence is explained in \cite{Casini:ls}. Our main contribution here, expanding the results found in \cite{Casini:2019kex}, is that this new proof uses only the neutral algebra.

The last feature is crucial when considering theories with local symmetries. In this scenario and concentrating in four spacetime dimensions, we have Wilson and 't Hooft loops, located on rings, associated with the center of the group \cite{Casini:ls}. It is well-known that these algebras do not commute. An interesting COA arises by the algebra of Wilson loops $W^{R}$ in a ring $R$ and the algebra of 't Hooft loops $T^{R'}$ in the complementary ring $R'$. The COA diagram reads 
\begin{eqnarray}
\mathcal{O}_{R}\vee W_{R} & \overset{\varepsilon}{\longrightarrow} & \mathcal{O}_{R}\nonumber \\
\updownarrow\prime\! &  & \:\updownarrow\prime\\
\mathcal{O}_{R'} \vee W_{R'} & \overset{\varepsilon'}{\longleftarrow} & \mathcal{O}_{R'}\vee W_{R'}\vee T_{R'}\,.\nonumber 
\end{eqnarray}
The entropic certainty follows with an index equal to the dimension of the center \cite{Casini:ls}.

In the first case, the intertwiners are natural order parameters, while the twists can be thought of as disorder operators. In the second case, the Wilson loop is the order parameter, while the 't Hooft loop is the disorder one. The entropic certainty thus captures, in a quantitative manner, the interplay between order/disorder parameters, whose commutation relations are crucial for the characterization of quantum phases.

The previous observation suggests that the complementarity between order and disorder parameters can be always framed by such entropic certainty principles. New uncertainty relations can be derived by using monotonicity of the relative entropy under general quantum channels or algebra restrictions. In the general case,
\bea
&&S_{\mathcal{M}}\left(\mathcal{E}(\omega)\mid\mathcal{E}(\omega\circ\varepsilon)\right)+S_{\mathcal{N}'}\left(\mathcal{E}'(\omega)\mid\mathcal{E}'(\omega\circ\varepsilon')\right) \leq\log\lambda \, , \nn \\
&&S_{\mathcal{\tilde{M}}}\left(\omega\mid\omega\circ\varepsilon\right)+S_{\mathcal{\tilde{N}}'}\left(\omega\mid\omega\circ\varepsilon'\right)\leq\log\lambda \, ,
\eea
where $\mathcal{E}$ and $\mathcal{E}'$ are any two quantum channels, and $\tilde{\mathcal{M}}\subset \mathcal{M}$ and $\tilde{\mathcal{N}}'\subset\mathcal{N}'$. At the present time, we do not know what is the class of uncertainty relations that can be proven in this way. They certainly expand the uncertainty relations proved in \cite{PhysRevLett.108.210405} with monotonicity of relative entropy in unexplored directions. It would be interesting if all uncertainty relations could be derived from such entropic certainties. We leave this as an open problem for the near future. Further uncertainties arise in the case where the global state $\omega$ considered is not pure. These are considered in appendix B.

\section{Discussion and future prospects.} 

In this letter, we have explored new implications of quantum complementarity. Most importantly we have found a formulation in which, instead of uncertainty principles, one naturally finds entropic certainty principles. Our main result is theorem \ref{cer_thm}. From this generic result, one obtains entropic uncertainty principles by using the monotonicity of the relative entropy. From a more physical perspective, our principle captures the interplay between order and disorder parameters in quantum theories. This has been described for QFTs with global and local symmetries.

Several open problems/questions are left for the future. The first concerns the consideration of more general quantum channels, instead of conditional expectations. We wonder if \eqref{cerp} holds for some a suitably defined $\lambda$ associated with the quantum channel. This might provide a path to extend the notion of index to more generic quantum channels. On the other hand, the validity of \eqref{cerp} only for conditional expectations could suggest they play a distinguished role in the description of quantum complementarity. The second problem concerns the extension of the proof to general infinite dimensional von Neumann algebras. Finally, a deeper understanding of how the entropic certainty discerns and characterize phases of quantum matter is needed. In particular, it would be interesting how \eqref{cerp} responses under dualities. Duality transformations typically map order parameters to disorder parameters, and we expect an interesting interplay when combining this feature with the entropic certainty.

\section{Acknowledgements.} We thank H. Casini and M. Huerta for guidance and many discussions. The work of J. M. is supported by the Simons Foundation through the ``It from Qubit'' collaboration. The work of D. P. is funded by CONICET, Argentina.

\bibliography{Bibliography}


\newpage
\section{APPENDIX}
\appendix
\section{A. The space of conditional expectations}\label{appx:rel}

Let $\mathcal{N}\subset\mathcal{M}\subset\mathcal{B}\left(\mathcal{H}\right)$
be an inclusion of algebras. A linear map $\varepsilon:\mathcal{M}\rightarrow\mathcal{N}$
is called a conditional expectation if it is positive, unital, and satisfies the bimodule property
\begin{equation}
\varepsilon\left(B_{1}\,A\,B_{2}\right)=B_{1}\varepsilon\left(A\right)B_{2}\,,\quad\forall A\in\mathcal{M}\textrm{ and }\forall B_{1},B_{2}\in\mathcal{N}\,.\label{ce_def_proprep}
\end{equation}
Conditional expectations are completely positive maps, and hence, special cases of quantum channels.

We want to characterize the space of all conditional expectations, denoted by $C\left(\mathcal{M},\mathcal{N}\right)$, for finite dimensional algebras. To such an end, we will use two building blocks giving by the following lemmas \ref{lemma_ce1} and \ref{lemma_ce2}.

\begin{lem} \label{lemma_ce1}
Let $\mathcal{N}\subset\mathcal{M}\subset\mathcal{B}\left(\mathcal{H}\right)$
be an inclusion of algebras, $\left\{ P_{1}^{\mathcal{N}},\ldots,P_{z_{\mathcal{N}}}^{\mathcal{N}}\right\} $ the minimal projectors of $\mathcal{Z}(\mathcal{N})$, and $\varepsilon\in C\left(\mathcal{M},\mathcal{N}\right)$. Let us define the algebras
\be
\mathcal{M}_{k}  :=  P_{k}^{\mathcal{N}}\mathcal{M}P_{k}^{\mathcal{N}}\,, \hspace{4mm} \mathcal{N}_{k} := P_{k}^{\mathcal{N}}\mathcal{N}P_{k}^{\mathcal{N}}\,.
\ee
Then, there exist unique conditional expectations $\varepsilon_{k}\in C\left(\mathcal{M}_{k},\mathcal{N}_{k}\right)$
such that
\be
\varepsilon(A)=\bigoplus_{k=1}^{z_{\mathcal{N}}}\varepsilon_{k}\left(P_{k}^{\mathcal{N}}AP_{k}^{\mathcal{N}}\right)\,.\label{ce_c1}
\ee
Conversely, given any set of conditional expectations
$\varepsilon_{k}\in C\left(\mathcal{M}_{k},\mathcal{N}_{k}\right)$,
the above formula defines a conditional expectation $\varepsilon\in C\left(\mathcal{M},\mathcal{N}\right)$.
\end{lem}
\begin{proof}
Let $\varepsilon\in C\left(\mathcal{M},\mathcal{N}\right)$
and $A\in\mathcal{M}$. Then, using the bimodule property of the conditional expectation we obtain
\begin{equation}
\varepsilon(A) = \varepsilon\left(\sum_{k,k'=1}^{z_{\mathcal{N}}}P_{k}^{\mathcal{N}}AP_{k'}^{\mathcal{N}}\right)= \sum_{k=1}^{z_{\mathcal{N}}}P_{k}^{\mathcal{N}}\varepsilon\left(P_{k}^{\mathcal{N}}AP_{k}^{\mathcal{N}}\right)P_{k}^{\mathcal{N}}\,.
\end{equation}
The last term naturally defines $\varepsilon_{k}:\mathcal{M}_{k}\rightarrow\mathcal{N}_{k}$ by means of
\begin{equation}
\varepsilon_{k}\left(P_{k}^{\mathcal{N}}AP_{k}^{\mathcal{N}}\right):=P_{k}^{\mathcal{N}}\varepsilon\left(P_{k}^{\mathcal{N}}AP_{k}^{\mathcal{N}}\right)P_{k}^{\mathcal{N}}\,.
\end{equation}
Using that $\varepsilon$ is a conditional expectation, a straightforward computation shows that $\varepsilon_{k}\in C\left(\mathcal{M}_{k},\mathcal{N}_{k}\right)$ for all $k=1,\ldots,z_{\mathcal{N}}$.

Conversely, given conditional expectations $\varepsilon_{k}\in C\left(\mathcal{M}_{k},\mathcal{N}_{k}\right)$ ($k=1,\ldots,z_{\mathcal{N}}$), it is easy to show
that $\varepsilon:\mathcal{M}\rightarrow\mathcal{N}$ defined as \eqref{ce_c1} is a conditional expectation in $ C\left(\mathcal{M},\mathcal{N}\right)$.
\end{proof}

Now, we need to consider the case when $\mathcal{N}$ is a factor.
In this case, we must have that
\be
\mathcal{N}  \cong  M_{n}\left(\mathbb{C}\right)\otimes\mathds{1}_{n'}\,,\hspace{4mm}
\mathcal{M} \cong M_{n}\left(\mathbb{C}\right)\otimes\mathcal{A}\,,\label{n_factor}
\ee
where $\mathcal{A}\subset M_{n'}(\mathbb{C})$ is some finite dimensional
subalgebra. In the most general case, we have that
\be
\mathcal{A}\cong\bigoplus_{j=1}^{z_{\mathcal{A}}}M_{a_{j}}\left(\mathbb{C}\right)\otimes\mathds{1}_{a'_{j}}\,,
\ee
and hence
\be
\mathcal{M}\cong\bigoplus_{j=1}^{z_{\mathcal{A}}}M_{n}\left(\mathbb{C}\right)\otimes M_{a_{j}}\left(\mathbb{C}\right)\otimes\mathds{1}_{a'_{j}}\cong\bigoplus_{j=1}^{z_{\mathcal{A}}}M_{n\times a_{j}}\left(\mathbb{C}\right)\otimes\mathds{1}_{a'_{j}}\,.\label{m_n_factor}
\ee
In this scenario, we have the following lemma.
\begin{lem}\label{lemma_ce2}
Let $\mathcal{N}\subset\mathcal{M}\subset\mathcal{B}\left(\mathcal{H}\right)$
as in (\ref{n_factor}-\ref{m_n_factor}). Then, any conditional expectation
$\varepsilon\in C\left(\mathcal{M},\mathcal{N}\right)$ is of the
form
\begin{equation}
\varepsilon\left(B\otimes A\right)=\varphi_{\varepsilon}\left(A\right)\,\left(B\otimes\mathbf{1}_{n'}\right)\,, \,\,\,\,\forall B\in M_{n}(\mathbb{C})\,,\;\forall A\in\mathcal{A}\,,\label{cex_factor}
\end{equation}
where $\varphi_{\varepsilon}$ is a state on $\mathcal{A}$. Conversely,
any state $\varphi_{\varepsilon}$ on $\mathcal{A}$ defines a conditional
expectation by means of \eqref{cex_factor} for simple elements, and
it is extended by linearity for more general ones.
\end{lem}
\begin{proof}
Let $\varepsilon\in C\left(\mathcal{M},\mathcal{N}\right)$, $A\in\mathcal{A}$
and $B\in M_{n}(\mathbb{C})$. Then, 
\bea
&\varepsilon\left(\mathbf{1}_{n}\otimes A\right)\cdot(B\otimes\mathbf{1}_{n'})=\varepsilon\left((\mathbf{1}_{n}\otimes A)\cdot(B\otimes\mathbf{1}_{n'})\right)=&\nonumber\\
&=\varepsilon\left((B\otimes\mathbf{1}_{n'})\cdot(\mathbf{1}_{n}\otimes A)\right)=(B\otimes\mathbf{1}_{n'})\cdot\varepsilon\left(\mathbf{1}_{n}\otimes A\right)&\nonumber\,,
\eea
which means that $\varepsilon\left(\mathbf{1}_{n}\otimes A\right)\in\mathcal{N}'$.
Since $\varepsilon\in C\left(\mathcal{M},\mathcal{N}\right)$, we
also have that $\varepsilon\left(\mathbf{1}_{n}\otimes A\right)\in\mathcal{N}$,
and hence $\varepsilon\left(\mathbf{1}_{n}\otimes A\right)\in\mathcal{Z}(\mathcal{N})=\mathds{1}_{n}\otimes\mathds{1}_{n'}$,
for all $A\in\mathcal{A}$. Then, there exists $\varphi_{\varepsilon}:\mathcal{A}\rightarrow\mathbb{C}$
such that
\be
\varepsilon\left(\mathbf{1}_{n}\otimes A\right)=\varphi_{\varepsilon}\left(A\right)\,\left(\mathbf{1}_{n}\otimes\mathbf{1}_{n'}\right)\,,\qquad\forall A\in\mathcal{A}\,,
\ee
and \eqref{cex_factor} automatically holds. To end, we have to show
that the map $\varphi_{\varepsilon}$ is a state on $\mathcal{A}$.
It is clear that $\varphi_{\varepsilon}$ is linear since $\varepsilon$
is linear. We also have that
\be
\varphi_{\varepsilon}\left(\mathbf{1}_{n'}\right)\,\left(\mathbf{1}_{n}\otimes\mathbf{1}_{n'}\right)=\varepsilon\left(\mathbf{1}_{n}\otimes\mathbf{1}_{n'}\right)=\mathbf{1}_{n}\otimes\mathbf{1}_{n'}\,,
\ee
and hence $\varphi_{\varepsilon}(\mathbf{1}_{n'})=1$. Finally, given $A\in\mathcal{A}$, we have
\bea
\varphi_{\varepsilon}\left(AA^{\dagger}\right)\,&&\left(\mathbf{1}_{n}\otimes\mathbf{1}_{n'}\right)=\varepsilon\left(\mathbf{1}_{n}\otimes AA^{\dagger}\right) \nonumber \\
&&=\varepsilon\left((\mathbf{1}_{n}\otimes A)\cdot(\mathbf{1}_{n}\otimes A)^{\dagger}\right)\geq0 \,,
\eea
which implies that $\varphi_{\varepsilon}\left(AA^{\dagger}\right)\geq0$
for all $A\in\mathcal{A}$.

Conversely, given a state $\varphi$ on $\mathcal{A}$, let us show
that \eqref{cex_factor} defines a conditional expectation $\varepsilon\in C\left(\mathcal{M},\mathcal{N}\right)$.
Given $A_{i}\in\mathcal{A}$, $B_{i}\in M_{n}(\mathbb{C})$, and $c_{i}\in\mathbb{C}$
($i=1,\ldots,N)$, then
\bea
\varphi \! \left(\sum_{i=1}^{N}c_{i}A_{i}\right)\!\left(B_{1}\otimes\mathbf{1}_{n'}\right) \!& = &\! \sum_{i=1}^{N}c_{i}\left[\varphi(A_{i})\!\left(B_{1}\otimes\mathbf{1}_{n'}\right)\right]\!, \hspace{6mm} \\
\varphi(A_{1})\!\left( \!\! \left(\sum_{i=1}^{N}c_{i}B_{i}\!\right)\!\otimes\!\mathbf{1}_{n'}\!\right) \!& = &\! \sum_{i=1}^{N}c_{i}\left[\varphi(A_{1})\!\left(B_{i}\otimes\mathbf{1}_{n'}\right)\right]\!,
\eea
which menas that the map $(B,A)\mapsto\varphi\left(A\right)\,\left(B\otimes\mathbf{1}_{n'}\right)$
is a bilinear map from the cartesian product $M_{n}(\mathbb{C})\times\mathcal{A}$
into $M_{n}(\mathbb{C})\otimes\mathds{1}_{n'}$, and hence defines
a unique linear map
\begin{equation}
\varepsilon\left(B\otimes A\right):=\varphi\left(A\right)\,\left(B\otimes\mathbf{1}_{n'}\right)
\end{equation}
from $M_{n}(\mathbb{C})\otimes\mathcal{A}$ into $M_{n}(\mathbb{C})\otimes\mathds{1}_{n'}$.
Given $B\in\mathcal{N}$ and $C\in\mathcal{M}$, then 
\bea
B&:=&\tilde{B}\otimes\mathbf{1}_{n'}\,, \hspace{7mm} \textrm{with }\tilde{B}\in M_{n}(\mathbb{C})\,,\\
C&:=&\sum_{i=1}^{N}\tilde{B}_{i}\otimes\tilde{A}_{i}\,, \hspace{2mm} \textrm{with }\tilde{B}_{i}\in M_{n}(\mathbb{C})\,,\;\tilde{A}_{i}\in\mathcal{A}\,.
\eea
This implies that,
\bea
&\varepsilon\left(C\cdot B\right)=\sum_{i=1}^{N}\varphi\left(\tilde{A}_{i}\right)\,\left((\tilde{B}_{i}\cdot\tilde{B})\otimes\mathbf{1}_{n'}\right)=&\nonumber\\= &\left[\sum_{i=1}^{N}\varphi\left(\tilde{A}_{i}\right)\,\left(\tilde{B}_{i}\otimes\mathbf{1}_{n'}\right)\right]\cdot\left(\tilde{B}\otimes\mathbf{1}_{n'}\right)=\varepsilon\left(C\right)\cdot B& \,. \nonumber
\eea
Finally, since any positive linear functional $\varphi:\mathcal{A}\rightarrow\mathbb{C}$
is completely positive, then the map
\be
\mathrm{id}_{n}\otimes\varphi:M_{n}(\mathbb{C})\otimes\mathcal{A}\rightarrow M_{n}(\mathbb{C})\otimes\mathbb{C}\cong M_{n}(\mathbb{C}) \, ,
\ee
given by 
\be
\left(\mathrm{id}_{n}\otimes\varphi\right)(B\otimes A):=\varphi(A)\,B \, ,
\ee
is positive. Furthermore, the inclusion map $\iota_{n,n'}:M_{n}(\mathbb{C})\rightarrow M_{n}(\mathbb{C})\otimes\mathds{1}_{n'}$
given by $\iota_{n,n'}(B)=B\otimes\mathbf{1}_{n'}$ is trivially positive.
Then, we have that
\be
\varepsilon=\iota_{n,n'}\circ\left(\mathrm{id}_{n}\otimes\varphi\right)\,,
\ee
is positive since it results from the composition of two positive linear transformations.
\end{proof}

Finally, combining lemmas \ref{lemma_ce1} and \ref{lemma_ce2}, we obtain lemma \ref{lemma_cem}. Our approach to the space of conditional expectations $C\left(\mathcal{M},\mathcal{N}\right)$ is more physically oriented, and certainly, it easies the proof of the theorem \ref{cer_thm}. The space $C\left(\mathcal{M},\mathcal{N}\right)$ was first studied long ago by Umegaki using different techniques \cite{umegaki1,umegaki2,umegaki3,umegaki4}. The final result provides a different parametrization of the space, but equivalent to ours.

\section{B. Proof of the entropic certainty principle (theorem \ref{cer_thm})}   \label{appx:proof}
In this appendix, we provide the proof of theorem \ref{cer_thm}. We bring here the equation for the convenience of the reader
\be
S_{\mathcal{M}}\left(\omega|\omega\circ\varepsilon\right)+S_{\mathcal{N}'}\left(\omega|\omega\circ\varepsilon'\right)=\log\lambda\,. \label{cer_eq_ap}
\ee
We start by reminding the representations \eqref{m_inc} for the algebras $\mathcal{N}\subset\mathcal{M}\subset\mathcal{B}\left(\mathcal{H}\right)$.
Automatically, we have for their commutants
\be
\hspace{-1.9mm} \mathcal{M}'  \cong  \bigoplus_{j=1}^{z_{\mathcal{M}}}\mathds{1}_{m_{j}}\otimes M_{m'_{j}}(\mathbb{C}), \hspace{2mm}
\mathcal{\mathcal{N}}' \cong  \bigoplus_{k=1}^{z_{\mathcal{N}}}\mathds{1}_{n_{k}}\otimes M_{n'_{k}}(\mathbb{C}).\label{m_inc_com}
\ee
To prove the theorem we compute each relative of \eqref{cer_eq_ap} separately. Let us focus on the first one
\be
S_{\mathcal{M}}\left(\omega|\omega\circ\varepsilon\right)=-S_{\mathcal{M}}\left(\omega\right)-\mathrm{Tr}_{\mathcal{M}}\left(\rho^{\omega}\log(\rho^{\omega\circ\varepsilon})\right)\,,\label{re_1}
\ee
where $\rho^{\omega}$ and $\rho^{\omega\circ\varepsilon}$ are the
density matrices of the states $\omega$ and $\omega\circ\varepsilon$ on the algebra $\mathcal{M}$. The von Neumann entropy $S_{\mathcal{M}}\left(\omega\right)$ in \eqref{re_1} will later cancel out with an equivalent term coming from the second relative entropy of \eqref{cer_eq_ap}. To analyze the other term, we need an expression for the density matrix $\rho^{\omega\circ\varepsilon}$. This density matrix is defined as the unique element in $\mathcal{M}$ satisfying
\be
\mathrm{Tr}_{\mathcal{M}}\left(\rho^{\omega\circ\varepsilon}A\right)=\omega\left(\varepsilon(A)\right)=\mathrm{Tr}_{\mathcal{M}}\left(\rho^{\omega}\varepsilon(A)\right) \, ,
\ee
for all $A \in \mathcal{M}$. According to lemma \ref{lemma_ce1}, we can decompose the conditional expectation $\varepsilon$ as in \eqref{ce_c1}
\bea
\mathrm{Tr}_{\mathcal{M}}\left(\rho^{\omega}\varepsilon\left(A\right)\right) &= & \mathrm{Tr}_{\mathcal{M}}\left(\rho^{\omega}\bigoplus_{k=1}^{z_{\mathcal{N}}}\varepsilon_{k}\left(A_{k}\right)\right) \nn \\ 
& = & \sum_{k=1}^{z_{\mathcal{N}}}\mathrm{Tr}_{\mathcal{\mathcal{M}}_{k}}\left[\rho_{k}^{\omega}\varepsilon_{k}\left(A_{k}\right)\right]\,, \label{ew_1}
\eea
where we have defined the operators $A_{k}:=P_{k}^{\mathcal{N}}\mathcal{M}P_{k}^{\mathcal{N}}$ and $\rho_{k}^{\omega}:=P_{k}^{\mathcal{N}}\rho^{\omega}P_{k}^{\mathcal{N}}$.
Let us now assume that
\be
\hspace{-2.5mm} A_{k}:=B_{k}\otimes\bigoplus_{j=1}^{z_{\mathcal{M}}}C_{kj}, \hspace{1.8mm} B_{k}\!\in \! M_{n_{k}}(\mathbb{C}),\,C_{kj} \! \in \! M_{\mu_{kj}}(\mathbb{C}).
\ee
Then, according to lemma \ref{lemma_ce2}, there exists states $\varphi_{k}^{\varepsilon}$
on $\bigoplus_{j=1}^{z_{\mathcal{M}}}M_{\mu_{kj}}(\mathbb{C})\cong\mathcal{M}_{k}\cap\mathcal{N}'$
such that
\be
\hspace{-2.5mm} \varepsilon_{k}\!\left(B_{k}\otimes\bigoplus_{j=1}^{z_{\mathcal{M}}}C_{kj}\right) \! =:\varphi_{k}^{\varepsilon}\left(\bigoplus_{j=1}^{z_{\mathcal{M}}}C_{kj}\right) \! \left(B_{k}\otimes\mathbf{1}_{n'_{k}}\right),\label{cex_eps}
\ee
where $n'_{k}:=\sum_{k=1}^{z_{\mathcal{M}}}\mu_{kj}$.\footnote{The numbers $n'_{k}$ defined in this way coincide with the ones introduced in section \ref{section_ce} once we have set $m'_{j}=1$, which
is correct since we are working in the canonical representation of $\mathcal{M}$.} Then, replacing \eqref{cex_eps} into \eqref{ew_1} we have that
\bea
&&\mathrm{Tr}_{\mathcal{M}}\left(\rho^{\omega}\varepsilon\left(A\right)\right)  =  \sum_{k=1}^{z_{\mathcal{N}}}\mathrm{Tr}_{\mathcal{\mathcal{M}}_{k}}\!\left[\rho_{k}^{\omega}\varepsilon_{k}\left(B_{k}\otimes\bigoplus_{j=1}^{z_{\mathcal{M}}}C_{kj}\right)\right]  \hspace{6mm} \nonumber \\
 && = \!\sum_{k=1}^{z_{\mathcal{N}}}\mathrm{Tr}_{\mathcal{\mathcal{N}}_{k}}[\mathrm{Tr}_{\mathcal{M}_{k}\cap\mathcal{N}'}\left(\rho_{k}^{\omega}\right)B_{k}]\,\mathrm{Tr}_{\mathcal{M}_{k}\cap\mathcal{N}'} \!\!\left(\!\rho_{k}^{\varepsilon}\bigoplus_{j=1}^{z_{\mathcal{M}}}C_{kj}\!\right) \!\!, 
\eea
where $\rho_{k}^{\varepsilon}\in\bigoplus_{j=1}^{z_{\mathcal{M}}}M_{\mu_{kj}}(\mathbb{C})$
are the density matrices corresponding to the states $\varphi_{k}^{\varepsilon}$ on $\mathcal{M}_{k}\cap\mathcal{N}'$, and $\mathrm{Tr}_{\mathcal{M}_{k}\cap\mathcal{N}'}(\rho_{k}^{\omega})\in M_{n_{k}}(\mathbb{C})$.
Then, it follows that
\be
\mathrm{Tr}_{\mathcal{M}}\left(\rho^{\omega}\varepsilon\left(A\right)\right)=\mathrm{Tr} _{\mathcal{\mathcal{M}}}\!\left[ \!  \left(\bigoplus_{k=1}^{z_{\mathcal{N}}}\mathrm{Tr}_{\mathcal{M}_{k}\cap\mathcal{N}'}(\rho_{k}^{\omega})\otimes\rho_{k}^{\varepsilon}\right) \! A\right]\!, \hspace{8mm}
\ee
which implies that
\be
\rho^{\omega\circ\varepsilon}=\bigoplus_{k=1}^{z_{\mathcal{N}}}\mathrm{Tr}_{\mathcal{M}_{k}\cap\mathcal{N}'}\left(\rho_{k}^{\omega}\right)\otimes\rho_{k}^{\varepsilon}\in\bigoplus_{k=1}^{z_{\mathcal{N}}}\mathcal{M}_{k}\subset\mathcal{M}\,.\label{rho_ce}
\ee
We now need to find the restriction of the state $\omega$ on $\mathcal{N}$, whose density matrix is denoted by $\rho^{\omega_{\mathcal{N}}}$. Given $A_{k}\in M_{n_{k}}(\mathbb{C})\simeq\mathcal{N}_{k}$, then
\bea
\hspace{-7mm} \mathrm{Tr}_{\mathcal{\mathcal{N}}}\left(\rho^{\omega_{\mathcal{N}}}\bigoplus_{k=1}^{z_{\mathcal{N}}}A_{k}\right)=\sum_{k=1}^{z_{\mathcal{N}}}\mathrm{Tr}_{\mathcal{\mathcal{N}}_{k}}\left[\mathrm{Tr}_{\mathcal{M}_{k}\cap\mathcal{N}'}\left(\rho_{k}^{\omega}\right)A_{k}\right]&&\nonumber\\
 =\mathrm{Tr}_{\mathcal{\mathcal{N}}}\left[\left(\bigoplus_{k=1}^{z_{\mathcal{N}}}\mathrm{Tr}_{\mathcal{M}_{k}\cap\mathcal{N}'}\left(\rho_{k}^{\omega}\right)\right) \! \left(\bigoplus_{k=1}^{z_{\mathcal{N}}}A_{k}\right)\right] && ,\label{so_N}
\eea
which implies that $\mathrm{Tr}_{\mathcal{M}_{k}\cap\mathcal{N}'}\left(\rho_{k}^{\omega}\right)=\rho_{k}^{\omega_{\mathcal{N}}}$.
Then, \eqref{rho_ce} becomes
\be
\rho^{\omega\circ\varepsilon}=\bigoplus_{k=1}^{z_{\mathcal{N}}}\rho_{k}^{\omega_{\mathcal{N}}}\otimes\rho_{k}^{\varepsilon}\,.\label{rho_ce2}
\ee
Now, we are in conditions to compute the second term of equation \eqref{re_1}. After some algebraic manipulations, we find
\bea
\hspace{-7mm} \mathrm{Tr}&&_{\mathcal{M}}\left(\rho \log \left(\rho^{\omega\circ\varepsilon}\right)\right)  \nn \\
&&=-S_{\mathcal{N}}\left(\omega\right)+\sum_{k=1}^{z_{\mathcal{N}}}\mathrm{Tr}_{\mathcal{M}_{k}\cap\mathcal{N}'}\left(\mathrm{Tr}_{\mathcal{N}_{k}}\left(\rho_{k}^{\omega}\right)\log\left(\rho_{k}^{\varepsilon}\right)\right), \hspace{5mm} 
\eea
where the operators $\mathrm{Tr}_{\mathcal{N}_{k}}\left(\rho_{k}^{\omega}\right)\in\bigoplus_{j=1}^{z_{\mathcal{M}}}M_{\mu_{kj}}(\mathbb{C})$. Finally, the relative entropy \eqref{re_1} can be written
\bea
 S_{\mathcal{M}} \! \left(\omega|\omega\circ\varepsilon\right)&=& -S_{\mathcal{M}}\left(\omega\right)+S_{\mathcal{N}}\left(\omega\right) \nn \\
&&- \! \sum_{k=1}^{z_{\mathcal{N}}}\mathrm{Tr}_{\mathcal{M}_{k}\cap\mathcal{N}'} \! \left(\mathrm{Tr}_{\mathcal{N}_{k}}\left(\rho_{k}^{\omega}\right)\log\left(\rho_{k}^{\varepsilon}\right)\right) . \hspace{7mm}\label{re_final}
\eea

The second relative entropy on the l.h.s. of \eqref{cer_eq_ap} can be computed in a similar way. In the end, we arrive to an equivalent expression as in \eqref{re_final} but for the inclusion of algebras $\mathcal{M}'\subset\mathcal{N}'$, where these algebras are as in \eqref{m_inc_com}.
In this case, we introduce the algebras
\bea
\mathcal{\mathcal{M}}'_{j} & := & P_{j}^{\mathcal{M}}\mathcal{\mathcal{M}}'P_{j}^{\mathcal{M}}\cong\mathds{1}_{m_{j}}\otimes M_{m'_{j}}(\mathbb{C})\,,\\
\mathcal{\mathcal{N}}'_{j} & := & P_{j}^{\mathcal{M}}\mathcal{\mathcal{N}}'P_{j}^{\mathcal{M}}\cong\left(\bigoplus_{k=1}^{z_{\mathcal{N}}}M_{\mu_{kj}}\left(\mathbb{C}\right)\right)\otimes M_{m'_{j}}(\mathbb{C})\,, \hspace{7mm}
\eea
where $P_{j}^{\mathcal{M}}$ ($j=1,\ldots,z_{\mathcal{M}}$) are the
minimal central projectors of $\mathcal{M}$. The density matrix
of the state $\omega$ on the algebra $\mathcal{N}'$ is denoted by
$\tilde{\rho}^{\omega}\in\bigoplus_{k=1}^{z_{\mathcal{N}}}M_{n'_{k}}(\mathbb{C})$,
and we define the operators 
\be
\tilde{\rho}_{j}^{\omega}:=P_{j}^{\mathcal{M}}\rho^{\omega}P_{j}^{\mathcal{M}}\in\mathcal{\mathcal{N}}'_{j}\,.\label{rho_j}
\ee
Again, using lemmas \ref{lemma_ce1} and \ref{lemma_ce2}, we have that
\bea
&& \varepsilon'=\bigoplus_{j=1}^{z_{\mathcal{M}}}\varepsilon'_{j}\,,\qquad\varepsilon'_{j}\in C\left(\mathcal{\mathcal{N}}'_{j},\mathcal{\mathcal{M}}'_{j}\right) , \\
&& \varepsilon'_{j} \! \left( \!\! \left(\bigoplus_{k=1}^{z_{\mathcal{N}}}C'_{kj} \! \right) \! \otimes \! B'_{j} \! \right) \! =\mathrm{Tr}_{\mathcal{N}'_{j}\cap\mathcal{M}} \! \left( \! \tilde{\rho}_{j}^{\varepsilon'} \! \bigoplus_{k=1}^{z_{\mathcal{N}}}C'_{kj}\right) \!\! \left(\mathbf{1}_{m_{j}} \! \otimes \! B'_{j}\right), \nn
\eea
for all $C'_{kj}\in M_{\mu_{kj}}\left(\mathbb{C}\right)$ and $B'_{j}\in M_{m'_{j}}(\mathbb{C})$,
where we have that $m_{j}=\sum_{k=1}^{z_{\mathcal{N}}}\mu_{kj}$ for the consistency of the inclusion. The density matrices $\tilde{\rho}_{j}^{\varepsilon'}$ represent states on $\mathcal{N}'_{j}\cap\mathcal{M}\cong\bigoplus_{k=1}^{z_{\mathcal{N}}}M_{\mu_{kj}}\left(\mathbb{C}\right)$.

Following algebraic manipulations analogous to the ones applied to the first term, the second term on the l.h.s. of \eqref{cer_eq_ap} reads
\bea
\hspace{-7mm} S_{\mathcal{N}'} \! \left(\omega|\omega\circ\varepsilon'\right)\!&=&\! -S_{\mathcal{N}'}\left(\omega\right)+S_{\mathcal{M}'}\left(\omega\right) \nn \\
&& \! - \! \sum_{j=1}^{z_{\mathcal{M}}}\mathrm{Tr}_{\mathcal{N}'_{j}\cap\mathcal{M}} \! \left(\mathrm{Tr}_{\mathcal{M}'_{j}} \! \left(\tilde{\rho}_{j}^{\omega}\right)\log\left(\tilde{\rho}_{j}^{\varepsilon'}\right) \! \right) \! .\label{re_final_2}
\eea
We now conclude that
\bea
\hspace{-7mm} S_{\mathcal{M}}\left(\omega|\omega\circ\varepsilon\right)&+&S_{\mathcal{N}'}\left(\omega|\omega\circ\varepsilon'\right) =\!\!\! \nn \\
 && \hspace{-5mm} -\sum_{k=1}^{z_{\mathcal{N}}}\mathrm{Tr}_{\mathcal{M}_{k}\cap\mathcal{N}'}\left(\mathrm{Tr}_{\mathcal{N}_{k}}\left(\rho_{k}^{\omega}\right)\log\left(\rho_{k}^{\varepsilon}\right)\right) \nn \\
 && \hspace{-5mm} -\sum_{j=1}^{z_{\mathcal{M}}}\mathrm{Tr}_{\mathcal{N}'_{j}\cap\mathcal{M}}\left(\mathrm{Tr}_{\mathcal{M}'_{j}}\left(\tilde{\rho}_{j}^{\omega}\right)\log\left(\tilde{\rho}_{j}^{\varepsilon'}\right)\right),\label{casi_cer}
\eea
where we have used that $S_{\mathcal{M}}\left(\omega\right)=S_{\mathcal{M}'}\left(\omega\right)$
and $S_{\mathcal{N}}\left(\omega\right)=S_{\mathcal{N}'}\left(\omega\right)$
since the state $\omega$ is pure in the global algebra $\mathcal{B}(\mathcal{H})$.

In order to simplify the r.h.s. of equation \eqref{casi_cer}, we notice that
\bea
\hspace{-7mm} \mathrm{Tr}_{\mathcal{N}_{k}}\left(\rho_{k}^{\omega}\right)\in\bigoplus_{j=1}^{z_{\mathcal{M}}}M_{\mu_{kj}}(\mathbb{C}) & \Rightarrow & \mathrm{Tr}_{\mathcal{N}_{k}}\left(\rho_{k}^{\omega}\right)=\bigoplus_{j=1}^{z_{\mathcal{M}}}\rho_{kj}^{\omega}\,,\\
\hspace{-7mm} \mathrm{Tr}_{\mathcal{M}'_{j}}\left(\tilde{\rho}_{j}^{\omega}\right)\in\bigoplus_{k=1}^{z_{\mathcal{N}}} M_{\mu_{kj}}(\mathbb{C}) & \Rightarrow & \mathrm{Tr}_{\mathcal{M}'_{j}}\left(\tilde{\rho}_{j}^{\omega}\right)=\bigoplus_{k=1}^{z_{\mathcal{N}}}\tilde{\rho}_{jk}^{\omega}\,.
\eea
Besides, a straightforward computation, like the one we did in equation \eqref{so_N},
shows that $\rho_{kj}^{\omega}=\tilde{\rho}_{jk}^{\omega}$ and 
\begin{equation}
\bigoplus_{j=1}^{z_{\mathcal{M}}}\bigoplus_{k=1}^{z_{\mathcal{N}}}\rho_{kj}^{\omega}\in\bigoplus_{j=1}^{z_{\mathcal{M}}}\bigoplus_{k=1}^{z_{\mathcal{N}}}M_{\mu_{kj}}(\mathbb{C})\,,
\end{equation}
is the density matrix of the state $\omega$ on the algebra
$\mathcal{N}'\cap\mathcal{M}$.

We also remind that the density matrices $\rho_{k}^{\varepsilon}$
and $\tilde{\rho}_{j}^{\varepsilon'}$, which define the conditional expectation $\varepsilon$ and $\varepsilon '$, can be written conveniently as
\be
\rho_{k}^{\varepsilon} = \bigoplus_{j=1}^{z_{\mathcal{M}}}p_{kj}^{\varepsilon}\,\rho_{kj}^{\varepsilon}\,, \hspace{5mm} 
\tilde{\rho}_{j}^{\varepsilon'} = \bigoplus_{k=1}^{z_{\mathcal{N}}}\tilde{p}_{jk}^{\varepsilon'}\,\tilde{\rho}_{jk}^{\varepsilon'}\,,
\ee
where $\rho_{kj}^{\varepsilon}$ and $\tilde{\rho}_{jk}^{\varepsilon'}$ are normalized and the leftover probabilities add up to one. Substituting these previous expressions and doing some algebraic manipulations, \eqref{casi_cer} becomes
\bea
S&&_{\mathcal{M}}\left(\omega|\omega\circ\varepsilon\right)+S_{\mathcal{N}'}\left(\omega|\omega\circ\varepsilon'\right)= \nn \\ 
&& \hspace{-1mm} -\sum_{j,k}\mathrm{Tr}_{M_{\mu_{kj}}(\mathbb{C})} \! \left[\rho_{jk}^{\omega}\left(\log\left(p_{kj}^{\varepsilon}\,\rho_{kj}^{\varepsilon}\right) \!+\! \log\left(\tilde{p}_{jk}^{\varepsilon'}\,\tilde{\rho}_{jk}^{\varepsilon'}\right)\right)\right] \!. \hspace{7mm} \label{casi_cer_2}
\eea
Given this last expression, if expression \eqref{cer_eq_ap} holds, it must exist a positive number $\lambda>0$, independent of $j$ and $k$, such that
\be
\log\left(p_{kj}^{\varepsilon}\,\rho_{kj}^{\varepsilon}\right)+\log\left(\tilde{p}_{jk}^{\varepsilon'}\,\tilde{\rho}_{jk}^{\varepsilon'}\right)=-\log(\lambda)\cdot\mathbf{1}_{\mu_{kj}}\,,\label{vinculo}
\ee
for all $j=1,\ldots,z_\mathcal{M}$ and $k=1,\ldots,z_\mathcal{N}$. To complete the proof, we must show that given $\varepsilon\in\hat{C}(\mathcal{M},\mathcal{N})$, parametrized by
$p_{kj}^{\varepsilon}$ and $\rho_{kj}^{\varepsilon}$,
we can choose $\tilde{p}_{jk}^{\varepsilon'}$ and $\tilde{\rho}_{jk}^{\varepsilon'}$, parametrized by $\varepsilon'\in\hat{C}(\mathcal{M},\mathcal{N})$, satisfying \eqref{vinculo}.

Let  $ \{ t_{l,jk} \}$ and $\{ \tilde{t}_{l,jk} \}$ be the eigenvalues of $\rho_{kj}^{\varepsilon}$ and $\tilde{\rho}_{jk}^{\varepsilon'}$ respectively ($l=1,\ldots,\mu_{kj}$). For \eqref{vinculo} to hold, the operator $\tilde{\rho}_{jk}^{\varepsilon'}$ must be diagonal in the same basis as $\rho_{kj}^{\varepsilon}$. Then, we must have
\be
p_{kj}^{\varepsilon}\,\tilde{p}_{jk}^{\varepsilon'}\,\tilde{t}_{l,jk}\,t_{l,jk}=\frac{1}{\lambda}\,,\hspace{5mm} \forall l=1,\ldots,\mu_{kj}\,.\label{det_ep_1}
\ee
This implies
\be
\hspace{-1.7mm}p_{kj}^{\varepsilon}\,\tilde{p}_{jk}^{\varepsilon'}=p_{kj}^{\varepsilon}\,\tilde{p}_{jk}^{\varepsilon'}\sum_{l=1}^{\mu_{kj}}\tilde{t}_{l,jk}=\frac{1}{\lambda}\sum_{l=1}^{\mu_{kj}}\frac{1}{t_{l,jk}}\equiv\frac{1}{\lambda}\cdot\lambda_{jk}(\varepsilon)\,.\label{det_ep_2}
\ee
Imposing
\be
1=\sum_{k=1}^{z_{\mathcal{N}}}\tilde{p}_{jk}^{\varepsilon'}=\frac{1}{\lambda}\sum_{k=1}^{z_{\mathcal{N}}}\frac{\lambda_{jk}}{p_{kj}^{\varepsilon}}\,,
\ee
This is only posible when $\sum_{k=1}^{z_{\mathcal{N}}}\frac{\lambda_{jk}}{p_{kj}^{\varepsilon}}\equiv\lambda$ is independent of the subindex $j$. The set of conditional expectations satisfying this condition was called $\hat{C}\left(\mathcal{M},\mathcal{N}\right)\subset C\left(\mathcal{M},\mathcal{N}\right)$ in the main text. In this case, the theorem is satisfied provided the dual conditional expectation is parametrized by
\be
\tilde{p}_{jk}^{\varepsilon'}:=\frac{1}{\lambda}\frac{\lambda(\varepsilon_{jk})}{p_{kj}^{\varepsilon}} \,, \hspace{6mm} 
\tilde{t}_{l,jk}:=\frac{1}{\lambda}\frac{1}{p_{kj}^{\varepsilon}\,\tilde{p}_{jk}^{\varepsilon'}\,t_{l,jk}}\,.\label{eprime_1}
\ee
A straightforward computation shows
\be
\lambda(\varepsilon')=\lambda(\varepsilon)=\lambda\,.
\ee

To end this appendix, we generalize \eqref{cer_eq_ap} to non-connected inclusions of algebras. In this case, $\mathcal{N} \subset \mathcal{M}$ can be uniquely decomposed as
\be
\mathcal{M}=\bigoplus_{i=1}^{z_c}  \mathcal{M}_i\, , \hspace{4mm} \mathcal{N}=\bigoplus_{i=1}^{z_c}  \mathcal{N}_i\,, \hspace{4mm} \mathcal{N}_i \subset \mathcal{M}_i \, ,
\ee
where $\mathcal{M}_i:=E_i \mathcal{M} E_i$, $\mathcal{N}_i:=E_i \mathcal{N} E_i$, and $\{E_i\,:\,i=1,\ldots,z_c\}$ are the minimal projectors of the common center $\mathcal{Z}(\mathcal{M}) \cap \mathcal{Z}(\mathcal{N})$. A conditional expectation $\varepsilon \in C \left(\mathcal{M},\mathcal{N}\right)$ can be uniquely descomposed as $\varepsilon=\bigoplus_i\varepsilon_i$, where $\varepsilon_i \in C\left(\mathcal{M}_i,\mathcal{N}_i\right)$. Then, we define $\hat{C} \left(\mathcal{M},\mathcal{N}\right)$ to be set formed by all conditional expectations $\varepsilon=\bigoplus_i\varepsilon_i$ such that $\varepsilon_i \in \hat{C} \left(\mathcal{M}_i,\mathcal{N}_i\right)$ for all $i=1,\ldots,z_c$.
\begin{cor} \label{cer_cor}
Let $\mathcal{N}\subset\mathcal{M}\subset\mathcal{B}\left(\mathcal{H}\right)$
be a general inclusion of finite dimensional algebras. Then, for every $\varepsilon\in\hat{C}\left(\mathcal{M},\mathcal{N}\right)$
there exists a unique $\varepsilon'\in\hat{C}\left(\mathcal{M},\mathcal{N}\right)$ such that
\be
S_{\mathcal{M}}\left(\omega|\omega\circ\varepsilon\right)+S_{\mathcal{N}'}\left(\omega|\omega\circ\varepsilon'\right)= \sum_{i=1}^{z_c} \omega(E_i)  \log\left(\lambda_i\right)\,,\label{cer_eq_cor}
\ee
holds for any global pure state $\omega$ on $\mathcal{B}\left(\mathcal{H}\right)$, and where $\lambda_i:=\lambda(\varepsilon_i)\equiv\lambda(\varepsilon'_i)$ is the algebraic index of the ``partial" conditional expectation $\varepsilon_i$.
\end{cor}
\begin{proof}
The proof consists of splitting the relative entropies on \eqref{cer_eq_cor} into relative entropies on the algebras  $\mathcal{M}_i$ and $\mathcal{N}'_i$, and apply theorem \ref{cer_thm} to each connected inclusion $\mathcal{N}_i \subset \mathcal{M}_i$.
\end{proof}

\section{C. The algebraic index of inclusion of algebras}\label{appx:index}
Given an inclusion of von Neumann type II$_{1}$ subfactors $\mathcal{N}\subset\mathcal{M}\subset\mathcal{B}(\mathcal{H})$,
Jones proposed \cite{Jones1983} an algebraic index $[\mathcal{M}:\mathcal{N}]\geq1$,
which ``measures'', in a certain sense, the size of $\mathcal{N}$ inside $\mathcal{M}$. The generalization to all types of algebras was developed independently by Kosaki and Longo \cite{KOSAKI1986123,longo1989}. They found that the algebraic index is most naturally attached to a conditional expectation $\varepsilon\in C(\mathcal{M},\mathcal{N})$, instead of an inclusion of subfactors $\mathcal{N}\subset\mathcal{M}\subset\mathcal{B}(\mathcal{H})$.

To explain the definition, let us first describe the space of weights $P(\mathcal{M},\mathcal{N})$.
A \textit{weight} $\eta\in P(\mathcal{M},\mathcal{N})$ is an unbounded
(and unnormalized) positive map $\eta:\mathcal{M}\rightarrow\mathcal{N}$
with dense domain in $\mathcal{M}_{+}$ (the positive subspace of
$\mathcal{M}$) satisfying the bimodule property \eqref{ce_def_prop}.
In particular, we have that $C(\mathcal{M},\mathcal{N})\subset P(\mathcal{M},\mathcal{N})$. Connes established a
canonical bijection between $P(\mathcal{M},\mathcal{N})$ and $P(\mathcal{N}',\mathcal{M}')$, see \cite{connes}. However, this bijection, in general, does not map $C(\mathcal{M},\mathcal{N})$
into $C(\mathcal{N}',\mathcal{M}')$. Still, what it remains true is that for a conditional expectation
$\varepsilon\in C(\mathcal{M},\mathcal{N})$, with $\epsilon^{-1}$ being the previous canonical Connes inverse (not neccessarily a conditional expectation itself), we have that $A\varepsilon^{-1}\left(\mathbf{1}\right)A^{\dagger}=\varepsilon^{-1}\left(\mathbf{1}\right)$
for any unitary $A\in\mathcal{M}$, and hence, $\varepsilon^{-1}\left(\mathbf{1}\right)\in\mathcal{Z}(\mathcal{M})$
whenever it is finite.

Let us now start by assuming $\mathcal{M}$ is a factor. In this case
\begin{equation}
\varepsilon^{-1}\left(\mathbf{1}\right)=\lambda(\varepsilon)\cdot\mathbf{1}\,,\hspace{6mm}1\leq\lambda(\varepsilon)\leq+\infty\,.
\end{equation}
The number $\lambda(\varepsilon)$ is called the index of
the conditional expectation $\varepsilon$. This definition is due to Kosaki \cite{KOSAKI1986123}. If there
exists $\varepsilon\in C(\mathcal{M},\mathcal{N})$ such $\lambda(\varepsilon)<+\infty$,
we say that $\mathcal{N}\subset\mathcal{M}$ is an inclusion of factors of \textit{finite index}. In this case, it was shown in \cite{longo1989} that there exists a unique conditional expectation $\varepsilon_{0}\in C(\mathcal{M},\mathcal{N})$
such that
\be
\lambda(\varepsilon_{0})=\min\left\{ \lambda(\varepsilon)\,:\,\varepsilon\in C(\mathcal{M},\mathcal{N})\right\} \,.
\ee
The number $\lambda(\varepsilon_{0})$ coincides with
the algebraic index $[\mathcal{M}:\mathcal{N}]$ defined by Jones whenever the relative commutant $\mathcal{N}'\cap\mathcal{M} $ is a factor. In such cases, the conditional expectational that minimize the index is the one that preserves the trace on $\mathcal{M}$. In the most general case, they do not coincide. If $\mathcal{N}\subset\mathcal{M}$ is not of a finite index, we simply
set $[\mathcal{M}:\mathcal{N}]=+\infty$. It was also shown by Kosaki \cite{KOSAKI1986123} that this definition implies, and indeed is equivalent to, the Pimsner-Popa bound \eqref{popa_eq}.

For an inclusion of finite index, given $\varepsilon\in C(\mathcal{M},\mathcal{N})$
we can define the dual conditional expectation $\varepsilon'\in C(\mathcal{N}',\mathcal{M}')$ by means of
\be
\varepsilon'(\cdot):=\frac{1}{\lambda(\varepsilon)}\varepsilon^{-1}(\cdot)\,.\label{dual_ce}
\ee
Notice that $\lambda(\varepsilon)=\lambda(\varepsilon')$ for all $\varepsilon\in C(\mathcal{M},\mathcal{N})$ with finite index.

Probably the simplest example to illustrate these ideas is the case of an inclusion of finite dimensional factors
\bea
\mathcal{N} & := & M_{n}(\mathbb{C})\otimes\mathbf{1}_{d}\otimes\mathbf{1}_{m'}\,,\\
\mathcal{M} & := & M_{n}(\mathbb{C})\otimes M_{d}(\mathbb{C})\otimes\mathbf{1}_{m'}\,.
\eea
According to the discussion on appendix \ref{appx:rel}, any $\varepsilon\in C(\mathcal{M},\mathcal{N})$ is determined by a state $\rho^{\epsilon}\in M_{d}(\mathbb{C})$ on $\mathcal{N}'\cap\mathcal{M}\simeq M_{d}(\mathbb{C})$ by means of
\be
\varepsilon\left(A\otimes B\otimes\mathbf{1}_{m'}\right)=\mathrm{Tr}_{M_{d}(\mathbb{C})}\left(\rho^{\epsilon}B\right)\left(A\otimes\mathbf{1}_{d}\otimes\mathbf{1}_{m'}\right).
\ee
As shown in \cite{longo1989}, the index of $\varepsilon$ is
\be
\lambda(\varepsilon)=\sum_{j=1}^{d}\frac{1}{t_{j}}\,,\label{index_finite_factors}
\ee
where $t_{j}$ are the eigenvalues of $\rho^{\epsilon}$. 
If $\rho^{\epsilon}$ is not invertible, then $\lambda(\varepsilon)=+\infty$.
A straightforward computation shows that the conditional expectation
$\varepsilon_{0}$ that minimizes the index is the one having all
equal eigenvalues $t_{j}=1/d$. In this case,
\be
\lambda(\varepsilon_{0})=[\mathcal{M}:\mathcal{N}]=d^{2}\,.
\ee

The more general case of an inclusion of algebras with centers was developed in \cite{teruya}. Let $P_{j}^{\mathcal{M}}$ ($j=1,\ldots,z_{\mathcal{M}}$) be the minimal projectors of $\mathcal{Z}(\mathcal{M})$. According
to our discussion above, we must have that
\be
\Lambda(\varepsilon):=\varepsilon^{-1}\left(\mathbf{1}\right)=\sum_{j=1}^{z_{\mathcal{M}}}c_{j}\,P_{j}^{\mathcal{M}}\,,\quad0\leq c_{j}\leq+\infty\,.
\ee
In this case, the index is an operator belonging
to the center of $\mathcal{M}$. According to the characterization
of the space $C(\mathcal{M},\mathcal{N})$ discussed in appendix \ref{appx:rel}, we have that \cite{teruya}
\be
c_{j} \equiv \sum_{k=1}^{z_{\mathcal{N}}} p_{jk}^{-1}\,\lambda(\varepsilon_{jk})\,,\label{index_op}
\ee
where $\lambda(\varepsilon_{jk})$ are the ``partial'' indices of
the conditional expectations $\varepsilon_{jk}\in C(\mathcal{M}_{jk},\mathcal{N}_{jk})$.
Notice that since $\mathcal{N}_{jk}$ and $\mathcal{M}_{jk}$ are
factors, $\lambda(\varepsilon_{jk})$ is a number.
In the case where $\mathcal{M}_{jk}$ is finite dimensional (and hence
$\mathcal{N}_{jk}$), $\lambda(\varepsilon_{jk})$ can be computed
according \eqref{index_finite_factors}. To obtain a number from \eqref{index_op},
we can take its operator norm
\be
\lambda(\varepsilon):=\left\Vert \Lambda(\varepsilon)\right\Vert =\max_{j=1,\ldots,z_{\mathcal{M}}}\left\{ c_{j}\right\} \,.\label{scalar_index}
\ee
In the particular case when all the constants \eqref{index_op} are independent of $j=1,\ldots,z_{\mathcal{M}}$, we have that
\bea
\Lambda(\varepsilon) & = & \lambda(\varepsilon)\cdot\mathbf{1}_{\mathcal{H}}\,,\label{good_ce}\\
\lambda(\varepsilon) & = & c_{j}=\sum_{k=1}^{z_{\mathcal{N}}}p_{jk}^{-1}\,\lambda(\varepsilon_{jk})\,,\hspace{4mm} \forall j=1,\ldots,z_{\mathcal{M}}\,.\label{good_ce2}
\eea
i.e. the index is a scalar. Furthermore, if $\lambda(\varepsilon)<+\infty$,
we can define the dual conditional expectation $\varepsilon'\in C\left(\mathcal{N}',\mathcal{M}'\right)$
as in \eqref{dual_ce}. It follows that $\lambda(\varepsilon)=\lambda(\varepsilon')$
and $\varepsilon'$ also satisfies the previous equation \eqref{good_ce}. In
these lines, it is useful to define
\be
\hat{C} \! \left(\mathcal{M},\mathcal{N}\right):=\left\{ \varepsilon \! \in \!  C \! \left(\mathcal{M},\mathcal{N}\right):\lambda(\varepsilon) \! < \! +\infty , \,\varepsilon\textrm{ satisfies }\eqref{good_ce}\right\} . \nn
\ee
The following key theorem, proven in \cite{teruya}, shows that $\hat{C}\left(\mathcal{M},\mathcal{N}\right)$ is non-empty and it contains the conditional expectation that minimizes the index norm \eqref{scalar_index}. 
\begin{thm}
Let $\mathcal{N}\subset\mathcal{M}\subset\mathcal{B}(\mathcal{H})$
be a connected inclusion,  $\varepsilon\in C\left(\mathcal{M},\mathcal{N}\right)$
a conditional expectation with finite norm index \eqref{scalar_index},
and $\varepsilon_{jk}\in C(\mathcal{M}_{jk},\mathcal{N}_{jk})$ the
``partial'' conditional expectations determined by $\varepsilon$.
Then, there exists $\tilde{\varepsilon}\in\hat{C}\left(\mathcal{M},\mathcal{N}\right)$ (with scalar index),
having the same ``partial'' conditional expectations $\varepsilon_{jk}\in C(\mathcal{M}_{jk},\mathcal{N}_{jk})$, such that
\begin{equation}
\lambda(\tilde{\varepsilon})\leq\lambda(\varepsilon)\,.
\end{equation}
\end{thm}
This theorem takes us to the last following lemma.
\begin{lem}
Let $\mathcal{N}\subset\mathcal{M}\subset\mathcal{B}(\mathcal{H})$
be a connected inclusion of finite norm index \eqref{scalar_index}. Then, there exists a conditional
expectation $\varepsilon_{0}\in C\left(\mathcal{M},\mathcal{N}\right)$
such that
\begin{equation}
\lambda(\varepsilon_{0})=\min\left\{ \lambda(\varepsilon)\,:\,\varepsilon\in C\left(\mathcal{M},\mathcal{N}\right)\right\} =:[\mathcal{M}:\mathcal{N}]\,.
\end{equation}
Moreover, we have $\varepsilon_{0}\in\hat{C}\left(\mathcal{M},\mathcal{N}\right)$,
and its partial conditional expectations $\varepsilon_{0,jk}\in\hat{C}\left(\mathcal{M}_{jk},\mathcal{N}_{jk}\right)$
are the ones which mimimize the index for the inclusion of factors
$\mathcal{N}_{jk}\subset\mathcal{M}_{jk}$.
\end{lem}

\end{document}